\documentclass[journal,10pt,twocolumn,twoside]{IEEEtran}
\usepackage{amsmath,amsfonts,amssymb,amsbsy,bm,paralist,theorem,color}
\usepackage{graphicx,algorithmic,algorithm,relsize}
\usepackage{multicol}
\usepackage{multirow}
\usepackage{mathtools}
\usepackage[caption=false,font=small,labelfont=sf,textfont=sf]{subfig}
\usepackage[font=small,labelfont=bf]{caption}
\usepackage{subcaption}
\usepackage{cite}
\usepackage{cuted}

\newtheorem{rem}{Remark}

\newtheorem{prop}{Proposition}

\DeclarePairedDelimiter{\ceil}{\lceil}{\rceil}
\graphicspath{{fig/}}
\definecolor{orange}{RGB}{255,107,0}
\definecolor{green}{RGB}{0,160,20}

\begin{document}

\title{A Low-Complexity Placement Design of Pinching-Antenna Systems}
\author{Ximing Xie,~\IEEEmembership{Member,~IEEE}, Fang Fang,~\IEEEmembership{Senior Member,~IEEE}, Zhiguo Ding,~\IEEEmembership{Fellew,~IEEE}, \\ and Xianbin Wang,~\IEEEmembership{Fellew,~IEEE}

\thanks{Ximing Xie, Fang Fang and Xianbin Wang are with the Department of Electrical and Computer Engineering, and Fang Fang is also with the Department of Computer Science, Western University, London, ON N6A 3K7, Canada (e-mail: \{xxie269, fang.fang, xianbin.wang\}@uwo.ca).}
\thanks{Zhiguo Ding is with Khalifa University, Abu Dhabi, UAE, and the University of Manchester, Manchester, M1 9BB, UK. (e-mail:zhiguo.ding@ieee.org).}
\vspace{-0.5cm}}
\maketitle

\begin{abstract}
 Pinching-antenna systems have recently been proposed as a new candidate for flexible-antenna systems, not only inheriting the reconfiguration capability but also offering a unique feature: establishing line-of-sight links to mitigate large-scale path loss. However, sophisticated optimization of the placement of pinching antennas has very high complexity, which is challenging for practical implementation. This paper proposes a low-complexity placement design, providing the closed-form expression of the placement of pinching antennas, to maximize the sum rate of multiple downlink users. Orthogonal multiple access (OMA) and non-orthogonal multiple access (NOMA) are both investigated when the pinching-antenna system is only equipped with a single antenna and only the OMA case is studied when there are multiple antennas equipped by the pinching-antenna system. Simulation results indicate pinching-antenna systems can outperform conventional fixed-antenna systems and are more suitable for large service areas.
\end{abstract}


\section{Introduction} 
Flexible-antenna systems, e.g., reconfigurable intelligent surfaces, movable antennas, and stacked intelligent metasurfaces, have gained significant attention due to their ability to reconfigure wireless channels \cite{xie2021joint, zhu2023modeling, an2023stacked}. Flexible-antenna systems have been shown to outperform conventional fixed-antenna systems. However, the locations of the aforementioned flexible-antenna systems are usually fixed or limited within the wavelength scale, which have difficulties to combat large-scale path loss especially when the line of sight (LoS) link is unavailable \cite{ding2024flexible}. To address this challenge, pinching-antenna systems were first proposed by a demonstration carried out by DOCOMO in 2022 \cite{suzuki2022pinching}. As shown in Fig. \ref{system model}, the low-cost dielectric materials such as plastic pinches can be dynamically deployed on a long dielectric waveguide to create increased number of radiation points. By adjusting the placement of pinches, new LoS links can be established for users who previously lacked direct LoS connectivity. Therefore, the pinching-antenna system can realize the capability of reconfiguring wireless channels by adjusting the placement of pinching antennas on the dielectric waveguide. As a result, pinching-antenna systems are expected to be a groundbreaking revolution in wireless communications, especially for the Millimeter Wave (mmWave) and Terahertz (THz) scenarios, where lacing of LoS channels might be a major issue. \par

Wireless channel reconfiguration is the most important capability of flexible-antenna systems. Since the placement of pinching antennas on the dielectric waveguide directly affects wireless channels, the placement optimization of pinching antennas is crucial to achieve higher performance gain. In \cite{ding2024flexible}, a pinching-antenna system assisted downlink was designed considering both orthogonal multiple access (OMA) and non-orthogonal multiple access (NOMA). As an initial research, \cite{ding2024flexible} heuristically discussed how to determine the placement of pinching antennas and also pointed out that sophisticated optimization of antenna placements is required to achieve optimal performance. However, sophisticated optimization causes high computational complexity, which is challenging for practical implementation. \par

Motivated by this, this paper focuses on a low-complexity placement design, which provides the closed-form solutions for the placement of pinching antennas based on the given user locations. In particular, a sum rate maximization problem is investigated in three different scenarios, i.e., single pinching antenna assisted OMA networks, multiple pinching antennas assisted OMA networks, and single pinching antenna assisted NOMA networks. By solving the sum rate maximization problem, the closed-form solutions of the pinching antenna's placement are derived for three scenarios. It is worth pointing out that the case of multiple pinching antennas assisted NOMA networks is too complicated to derive a closed-form, hence, a sophisticated optimization algorithm is required, which is beyond the low-complexity scope of this paper and will be investigated in the future work. Simulation results indicate the considered pinching-antenna system yields higher sum rate compared to conventional fixed-antenna systems in all scenarios.  
\begin{figure}[t]
     \centering
     \includegraphics[width=0.44\textwidth]{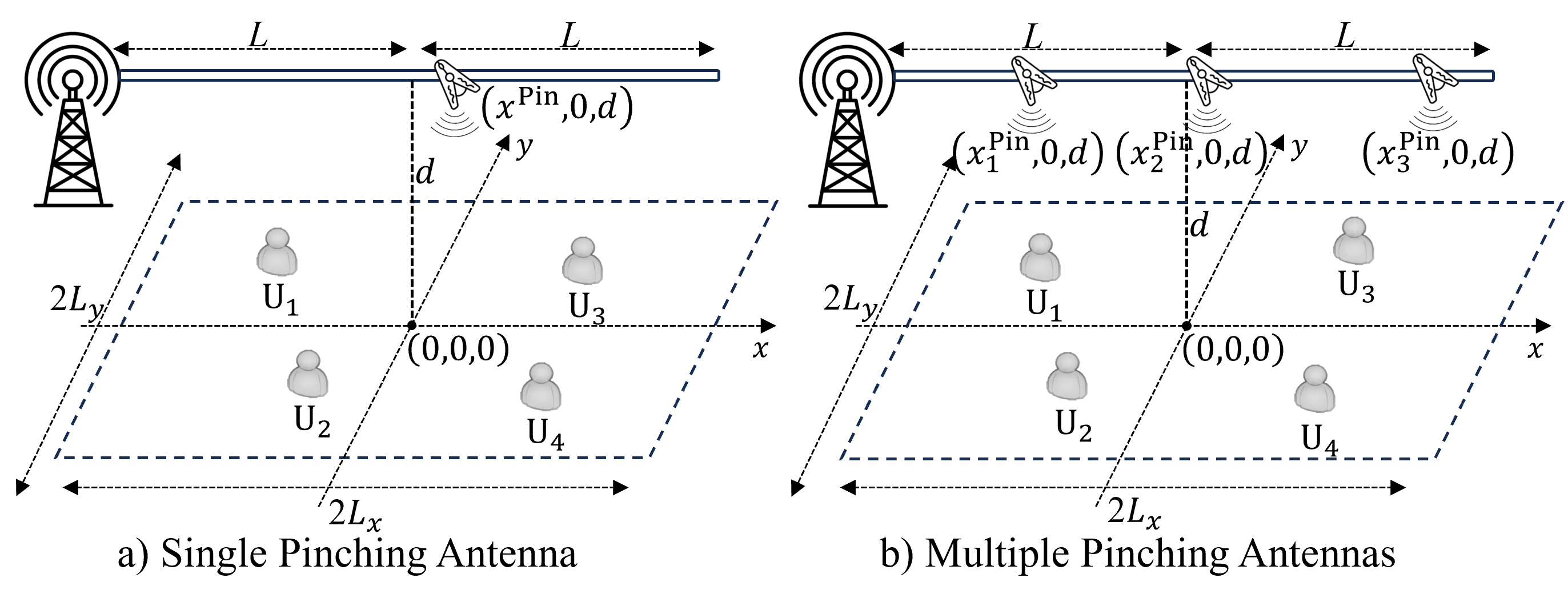}
     \caption{Illustration of two pinching-antenna systems: a) single pinching antenna on a waveguide and b) multiple pinching antennas on a waveguide.}
     \label{system model}
     \vspace{-0.7cm}
\end{figure}
\section{Pinching Antenna assisted OMA Networks}
In this section, time-division multiple access (TDMA) is considered as the multiple access technique. Since only one user is served during each time slot, the placement of pinching antennas is only related to one user's location in each time slot.
\subsection{A Single Pinching Antenna on a Waveguide}
Consider a TDMA-based downlink communication network, where the BS serves $M$ signal antenna users using a pinching antenna on a single waveguide. The $m$-th user is denoted by ${\rm U}_m$. To illustrate the positional relationship, we establish a three-dimensional Cartesian coordinate system. Without loss of generality, we assume that the waveguide has a length of $2L_x$ and is positioned parallel to the $x$-axis at the height of $d$. Meanwhile, $M$ users are randomly distributed within a rectangular service area located in the $x-y$ plane with side lengths of $2L_x$ and $2L_y$. In the TDMA system, it is assumed that there are $M$ time slots, and user ${\rm U}_m$ is served during the $m$-th time slot. The pinching antenna can be repositioned to different locations to serve distinct users in different time slots. Let $\psi_m^{\rm Pin} = (x_m^{\rm Pin}, 0, d)$ represent the location of the pinching antenna during the $m$-th time slot. Let $\psi_0 = (0,0,d)$ and $\psi_m = (x_m, y_m, 0)$ denote the locations of the center of the waveguide and ${\rm U}_m$, respectively. Therefore, ${\rm U}_m$'s data rate can be expressed as follows:
\begin{equation}
    R_m = \frac{1}{M} \log_2 \left(1+\frac{\eta P_m}{\left| \psi_m - \psi^{\rm Pin}_m  \right|^2 \sigma^2} \right), \label{data rate single OMA}
\end{equation}
where $\eta = \frac{c^2}{16 \pi^2 f_c^2}$, and where $c$ denotes the speed of light, $f_c$ denotes the carrier frequency, $P_m$ denotes the transmission power of ${\rm U}_m$'s signal and $\sigma^2$ denotes the power of additive white Gaussian noise (AWGN). It is assumed that the pinching antenna can be moved to a specific location on the waveguide in a short period of time, hence, the time taken for antenna movement is neglected. Then, the sum rate maximization problem can be formulated as:

\begin{subequations}\label{Prob1}
    \begin{align}
        {\rm P1}: &\max_{\{P_m, \psi_m^{\rm Pin}, \forall m\}}\;\; \sum\limits_{m=1}^M R_m \label{P10}\\
		&~\mathrm{s.t.}\qquad 0 \leq P_m \leq P_{max}, \forall m \label{P11}\\
            & \qquad \quad\; -L_x \leq x_m^{\rm Pin} \leq L_x, \forall m. \label{P12}
    \end{align}
\end{subequations}
Note that $R_m$ reaches to its optimal value when the transmission power $P_m$ is maximized, and the distance between ${\rm U}_m$ and the pinching antenna $\lvert \psi_m - \psi^{\rm Pin}_m  \rvert^2$ is minimized. As a result, the optimal transmission power is $P_m^* = P_{max}, \forall m$ and the optimal location of the pinching antenna is $\psi^{*{\rm Pin}}_m = (x_m,0,d), \forall m$.  
\subsection{Multiple Pinching Antennas on a Waveguide}
Assume there are $N$ active pinching antennas on a waveguide and the location of the $n$-th pinching antenna in the $m$-th time slot is denoted by $\psi_{n,m}^{\rm Pin}$. Since $N$ pinching antennas are on the same waveguide, they can be treated as a conventional linear array antennas. By adopting the spherical wave channel model, the channel vector of ${\rm U}_m$ is expressed as follows:
\begin{equation}
    \mathbf{h}_m = \left[\frac{\eta^{\frac{1}{2}} e^{-j \frac{2\pi}{\lambda} |\psi_m - \psi_{1,m}^{\rm Pin}|}}{|\psi_m - \psi_{1,m}^{\rm Pin}|} \; \cdots \; \frac{\eta^{\frac{1}{2}} e^{-j \frac{2\pi}{\lambda} |\psi_m - \psi_{N,m}^{\rm Pin}|}}{|\psi_m - \psi_{N,m}^{\rm Pin}|} \right]^T, \label{channel vector multi OMA}
\end{equation}
where $\lambda = \frac{c}{f_c}$ denotes wavelength of the carrier. According to \cite{ding2024flexible}, the signal vector $\mathbf{s}_m$ of ${\rm U}_m$ can be expressed as follows:
\begin{equation}
    \mathbf{s}_m = \sqrt{\frac{P_m}{N}} \left[e^{-j \theta_{1,m}} \; \cdots \; e^{-j \theta_{N,m}} \right]^T s_m,  \label{signal vector multi OMA}
\end{equation}
where $\theta_{n,m} = 2\pi \frac{|\psi_0^{\rm Pin} - \psi_{n.m}^{\rm Pin}|}{\lambda_g}$ represents the phase shift experienced by the signal at the $n$-th pinching antenna during the $m$-th time slot. $\psi_0^{\rm Pin}$ denotes the location of the feed point of the waveguide, and $\lambda_g$ denotes the waveguide wavelength in a dielectric waveguide. The waveguide wavelength $\lambda_g$ can be calculated as $\lambda_g = \frac{\lambda}{n_{\rm eff}}$, where $n_{\rm eff}$ is the effective refractive index of the dielectric waveguide. Therefore, ${\rm U}_m$'s received singal can be expressed as follows:
\begin{equation}
    \begin{split}
        &y_m = \mathbf{h}_m^H \mathbf{s}_m + w_m \\
            &= \left(\sum\limits_{n=1}^N \frac{\eta^{\frac{1}{2}} e^{-j\frac{2\pi}{\lambda} \left|\psi_m - \psi_{n,m}^{\rm Pin}\right|}}{\left|\psi_m - \psi_{n,m}^{\rm Pin}\right|} e^{-j \theta_{n,m}}\right) \sqrt{\frac{P_m}{N}} s_m + w_m, \label{received signal multi OMA}
    \end{split}
\end{equation}
where $w_m$ denotes the AWGN at ${\rm U}_m$. By assuming only ${\rm U}_m$ is served in the $m$-th time slot and the average power of $s_m$ is 1, ${\rm U}_m$'s data rate can be expressed as follows:
\begin{equation}
    \begin{split}
        &R_m =  \\
        &\frac{1}{M}\log_2\left(1 + \frac{P_m}{N\sigma^2}\left| \sum\limits_{n=1}^N \frac{\eta^{\frac{1}{2}} e^{-j\frac{2\pi}{\lambda} |\psi_m - \psi_{n,m}^{\rm Pin}|}}{|\psi_m - \psi_{n,m}^{\rm Pin}|} e^{-j \theta_{n,m}} \right|^2 \right). \label{data rate multi OMA}
    \end{split}
\end{equation}
Then, the sum rate maximization problem can be formulated as follows:
\begin{subequations}\label{Prob2}
    \begin{align}
        {\rm P2}: &\max_{\{P_m, \psi_{n,m}^{\rm Pin}, \forall n,m\}}\;\; \sum\limits_{m=1}^M R_m \label{P20}\\
		&~\mathrm{s.t.}\qquad 0 \leq P_m \leq P_{max}, \forall m \label{P21}\\
            & \qquad \quad\; -L_x \leq x_{n,m}^{\rm Pin} \leq L_x, \forall n,m. \label{P12}
    \end{align}
\end{subequations}
Note that $R_m$ increases monotonically with $P_m$, hence, the optimal transmission power is $P_m^* = P_{max}, \forall m$. \par

Recall TDMA, each user is served individually, hence, we can optimize each $R_m$ individually. By applying Cauchy-Schwarz inequality, the upper bound of \eqref{data rate multi OMA} can be expressed as \eqref{data rate upper bound multi OMA} on the bottom of this page. Note that $R_m$ can reach its upper bound when 
\begin{equation}
    \frac{2\pi}{\lambda}\left|\psi_m - \psi_{n,m}^{\rm Pin}\right| + \theta_{n,m} = 2k_n\pi, k_n \in \mathbf{N}, \forall n \label{upper bound condition multi OMA}
\end{equation}
holds. Note that ${\rm U}_m$'s data rate is increasing when the pinching antenna is closer to ${\rm U}_m$. To find the optimal $\psi_{n,m}^{\rm Pin}$, we formulate the following optimization problem:
\begin{strip}
    \vspace{-0.85cm}
    \begin{equation}
    R_m \leq \frac{1}{M} \log_2 \left(1 + \frac{P_m}{N\sigma^2} \sum\limits_{n=1}^N\frac{\eta}{ \left|\psi_m - \psi_{n,m}^{\rm Pin}\right|^2}\sum\limits_{n=1}^N \left|e^{-j\left(\frac{2\pi}{\lambda}\left|\psi_m - \psi_{n,m}^{\rm Pin}\right|+\theta_{n,m}\right)} \right|^2\right) \leq \frac{1}{M} \log_2 \left(1 + \frac{ P_m}{\sigma^2} \sum\limits_{n=1}^N\frac{\eta}{ \left|\psi_m - \psi_{n,m}^{\rm Pin}\right|^2}\right). \label{data rate upper bound multi OMA}
    \end{equation}
    \vspace{-0.5cm}
\end{strip}

\begin{subequations}\label{Prob3}
    \begin{align}
        {\rm P3}: &\min_{\{\psi_{n,m}^{\rm Pin}\}}\;\; \left|\psi_m - \psi_{n,m}^{\rm Pin}\right| \label{P30}\\
		&~\mathrm{s.t.}\qquad \frac{2\pi}{\lambda}\left|\psi_m - \psi_{n,m}^{\rm Pin}\right| + \theta_{n,m} = 2k_n\pi, k_n \in \mathbf{N} \label{P31}\\
            & \qquad \quad\; -L_x \leq x_{n,m}^{\rm Pin} \leq L_x. \label{P32}
    \end{align}
\end{subequations}

Let $\Tilde{\psi}_m^{\rm Pin}$ denote the closet position on the waveguide corresponding to $\psi_m$. Once $\psi_m = (x_m, y_m, 0)$ is determined, $\Tilde{\psi}_m^{\rm Pin} = (x_m, 0, d)$ is determined as well. However, $\psi_{n,m}^{\rm Pin} = \Tilde{\psi}_m^{\rm Pin}$ may not satisfy constraint \eqref{P31}, hence, we introduce an offset $\Delta_{n,m}^{\rm off}$ to describe the distance between $\psi_{n,m}^{\rm Pin}$ and $\Tilde{\psi}_m^{\rm Pin}$. As a result, $\psi_{n,m}^{\rm Pin} = (x_m + \Delta_{n,m}^{\rm off}, 0, d)$. The aim is to make $\psi_{n,m}^{\rm Pin}$ closer to $\Tilde{\psi}_m^{\rm Pin}$ as much as possible meanwhile satisfying constraint \eqref{P31}. Then, ${\rm P3}$ can be recast into
\begin{subequations}\label{Prob4}
    \begin{align}
        {\rm P4}: &\min_{\{\Delta_{n,m}^{\rm off}\}}\;\; \left|\Delta_{n,m}^{\rm off}\right| \label{P40}\\
		&~\mathrm{s.t.}\qquad \left(\left(\Delta_{n,m}^{{\rm off}}\right)^2 + y_m^2 + d^2 \right)^{\frac{1}{2}} \notag \\
            & \qquad \quad + n_{\rm eff}(L + x_m + \Delta_{n,m}^{\rm off})  = k_n\lambda, k_n \in \mathbf{N} \label{P41} \\
            & \qquad \quad -L_x - x_m \leq \Delta_{n,m}^{\rm off} \leq L_x - x_m \label{P42}.
    \end{align}
\end{subequations}
To solve ${\rm P4}$, we define the function $f(x) = (x^2 + D_1)^{\frac{1}{2}} + n_{\rm eff} (x + D_2)$, where $D_1 = y_m^2 + d^2$ and $D_2 = L + x_m$. We assume that the offset is non-negative (i.e. the pinching antenna always offsets to the direction opposite to the the feed point of the waveguide). In order to find the minimal $\Delta_{n,m}^{\rm off}$, we need to find the minimal $k_n$. Note that $f(x)$ is monotonically increasing with $x$ when $x \geq 0$, hence, we have $k\lambda = f(x) \geq f(0)$. Let $\delta = {\rm mod}(f(0),\lambda)$, where ${\rm mod}(a,b)$ denotes the modulo operation of $a$ by $b$. The minimal $k_n^*$ can be expressed as follows:
\begin{equation}
    k_n^* = \frac{f(0) - \delta}{\lambda} + \ceil[\Big]{\frac{\delta}{\lambda}} + (n-1) \label{optimal k multi OMA},
\end{equation}
where $\lceil a \rceil$ denotes the round up operation. Then, the minimal $\Delta_{n,m}^{*\rm off}$ can be obtained by solving the equation 
\begin{equation}
    \left(\left(\Delta_{n,m}^{*{\rm off}}\right)^2 + D_1\right)^{\frac{1}{2}} + n_{\rm eff}\left(D_2 + \Delta_{n,m}^{*\rm off}\right)  = k^*_n\lambda. \label{optimal equation multi OMA}
\end{equation}
\begin{prop}
    Equation \eqref{optimal equation multi OMA} always has a positive real solution.
\end{prop}
\begin{proof}
    To prove this proposition, we need to define a new function $g(x) =  \left(x^2 + D_1\right)^{\frac{1}{2}} + n_{\rm eff}\left(D_2 +x\right)  - k^*_n\lambda$. It is noted that $g(x)$ is continuous on the closed interval $[0, k_n^*\lambda]$. We have
    \begin{equation}
        \begin{split}
            g(0) &= D_1^{\frac{1}{2}} + n_{\rm eff}D_2 - k_n^*\lambda\\
                &\overset{(a)}{=} D_1^{\frac{1}{2}} + n_{\rm eff}D_2 - D_1^{\frac{1}{2}} +  n_{\rm eff}D_2 + \delta -  \ceil[\Big]{\frac{\delta}{\lambda}} \lambda - (n-1)\lambda\\
                &= \delta -  \ceil[\Big]{\frac{\delta}{\lambda}} \lambda - (n-1)\lambda \overset{(b)}{\leq} 0, \label{g(0)}
        \end{split}
    \end{equation}
where $(a)$ is obtained by substituting $k_n^*$ with \eqref{optimal k multi OMA} and $(b)$ can be proved by the fact that $\delta < \lambda$. We further have 
\begin{equation}
    g( k_n^*\lambda) =  \left((k_n^*\lambda))^2 + D_1\right)^{\frac{1}{2}} + n_{\rm eff}\left(D_2 +k_n^*\lambda\right)  - k^*_n \lambda \overset{(c)}{>} 0,
\end{equation}
where $(c)$ is due to $n_{\rm eff}$ is usually greater than one. Therefore, $g(0)g(k_n^*\lambda) \leq 0$. According to intermediate value theorem, $g(x)$ exists at least one real root in $[0, k_n^*\lambda]$. The proposition is proved. 
\end{proof}
\begin{prop}
    $\Delta_{n,m}^{*{\rm off}}$ is always less than $n\lambda$.
\end{prop}
\begin{proof}
    According to \eqref{optimal equation multi OMA}, we have $f\left(\Delta_{n,m}^{\rm off}\right) = k_n^*\lambda$. By substituting $k_n^*$ with \eqref{optimal k multi OMA}, we further have $f\left(\Delta_{n,m}^{\rm off}\right) = D_1^{\frac{1}{2}} + n_{\rm eff}D_2 - \delta + n\lambda$. Hence, $$f(n\lambda) - f\left(\Delta_{n,m}^{\rm off}\right) = (n^2\lambda^2 + D_1)^{\frac{1}{2}} - D_1^{\frac{1}{2}} + (n_{\rm eff} - 1)n\lambda + \delta$$ is greater than zero. As a result, $f(n\lambda) > f\left(\Delta_{n,m}^{\rm off}\right)$. Recall $f(x)$ is monotonically increasing with $x$ when $x \geq 0$, the proposition is proved. \label{prop2}
\end{proof}
According to Proposition 2, the offset of pinching antennas is within a few wavelength, which has a negligible impact on the distance $\left|\psi_m - \psi_{n,m}^{\rm Pin}\right|$, particularly when $f_c$ is very large. Therefore, by assuming $\frac{\left|\psi_m - \Tilde{\psi}_{m}^{\rm Pin}\right|}{\left|\psi_m - \psi_{n,m}^{\rm Pin}\right|} \approx 1, \forall n.$, the data rate of ${\rm U}_m$ can be simplified as follows:
\begin{equation}
    R_m \approx  \frac{1}{M} \log_2 \left(1 + \frac{N P_m \eta}{\sigma^2 \left|\psi_m - \Tilde{\psi}_{m}^{\rm Pin}\right|^2}\right).
\end{equation}
The location of the $n-$th pinching antenna in the $m-$th time slot can be obtained as
\begin{equation}
    \psi_{n,m}^{\rm Pin} = \Tilde{\psi}_m^{\rm Pin} + \Delta_{n,m}^{* {\rm off}}.
\end{equation}
\begin{rem}
    It is noted that equation \eqref{optimal equation multi OMA} can be reformulated as a quadratic polynomial. Consequently, the quadratic formula can be utilized to determine $\Delta_{n,m}^{* {\rm off}}$ efficiently with low-complexity.
\end{rem}
\section{Pinching Antenna assisted NOMA Networks}
A  characteristic of pinching-antenna systems is that the pinching antenna operating on a specific waveguide is required to transmit the same signal, which inspires the use of NOMA \cite{ding2024flexible}. Since all users are served simultaneously in NOMA networks, it is impractical to precisely adjust the position of the antennas to exclusively serve a single user. If we assume all users' locations to be fixed during one time period, an efficient solution is to place the pinching antenna at a preset position.
\vspace{-0.5cm}
\subsection{A Single Pinching Antenna on a Waveguide}
Since multiple users are served simultaneously in NOMA networks, the transmitted signal must be superimposed to accommodate the signals of all users. Let $s_m$ denote ${\rm U}_m$'s signal and $p_m$ denote the transmission power of $s_m$. The superimposed signal is given by $s = \sum_{m=1}^M \sqrt{p_m}s_m$. It is noted that the channel gain is determined by the location of the user and the position of the pinching antenna on the waveguide. Therefore, the channel gain of ${\rm U}_m$ is given by $h_m = \frac{\sqrt{\eta}}{\left| \psi_m - \psi^{\rm Pin}  \right|}$, where $\psi^{\rm Pin}$ denotes the position of the pinching antenna on the waveguide. According to the principle of power-domain NOMA, the user with the strong channel gain needs to decode the signals of the users with weaker channel gains. For example, if the channel gains are sorted as $|h_1|^2 \leq \cdots \leq |h_M|^2$, ${\rm U}_m$ needs to decode ${\rm U}_j$'s signal first, $1\leq j \leq m-1$ before decoding its own signal. According to the principle of SIC, users with stronger channel gains than ${\rm U}_m$ introduce interference during the decoding of  ${\rm U}_m$'s signal. Given the ordering of the channel gains for $M$ users, denoted by $\chi$, the interference set of ${\rm U}_m$, represented as $\mathcal{I}_m^{\chi}$, comprises all indices of users whose signals are treated as interference during the decoding of ${\rm U}_m$'s signal. In addition, users belonging to $\mathcal{I}_m^{\chi}$ will decode ${\rm U}_m$'s signal. Then, the data rate of ${\rm U}_i$ to decode ${\rm U}_m$'s signal can be expressed as follows:
\begin{equation}
    R_{i,m} = \log_2 \left(1+ \frac{|h_i|^2 p_m}{\sum_{j \in \mathcal{I}_m^{\chi}} |h_i|^2 p_j + \sigma^2} \right), i \in \mathcal{I}_m^{\chi}. \label{data rate im single NOMA}
\end{equation}
The achievable data rate of ${\rm U}_m$ is given by
\begin{equation}
    R_m = \min \{R _{m,m}, R_{i,m}, i \in \mathcal{I}_m^{\chi}\}. \label{achievable data rate single NOMA }
\end{equation}
Then, the sum rate maximization problem can be formulated as follows:
\begin{subequations}\label{Prob5}
    \begin{align}
        {\rm P5}: &\max_{\{ \mathbf{p}, \psi^{\rm Pin}\}}\;\; \sum\limits_{m=1}^M R_m \label{P50}\\
		&~\mathrm{s.t.} \sum\limits_{m=1}^M p_m \leq P_{max}, \forall m \label{P51}\\
            & \qquad R_m \geq R_t, \forall m \label{P52} \\
            & \qquad -L_x \leq x_{n,m}^{\rm Pin} \leq L_x, \forall n,m, \label{P53}
    \end{align}
\end{subequations}
where $\mathbf{p}$ denotes the power vector collecting all $p_m, \forall m$, and $R_t$ denotes the minimal target data rate. Constraint \eqref{P51} ensures that the total transmission power will not exceed the maximum power budget and \eqref{P52} guarantees the successful execution of SIC. It is noted that $\psi^{\rm Pin}$ affects the channel gain ordering $\chi$ and further affects the interference set $\mathcal{I}_m^{\chi}$. Moreover, two optimization variables are  coupled together. As a result, ${\rm P5}$ is non-convex and it is difficult to be solved in polynomial time. \par
 To efficiently solve ${\rm P5}$, we propose a two-phase optimization design: 1) Placement optimization for the pinching antenna; 2) Power allocation for each user.
\subsubsection{Placement Optimization for the Pinching Antenna}
Recall that the pinching antenna has the ability to build LoS link, the general free-space path loss for the LoS link of the $m$-th user can be denoted as
\begin{equation}
    L_m^{\rm LoS} = 20\log\left(\frac{4\pi f_c |\psi_m - \psi^{\rm Pin}|}{c}\right). \label{path loss single NOMA}
\end{equation}
Instead of directly optimizing non-convex \eqref{P50}, we can minimize the total path loss to find a sub-optimal placement of the pinching antenna. For a given transmission power $p_m$ of the $m$-th user, the placement optimization problem can be formulated as follows:
\begin{subequations}\label{Prob6}
    \begin{align}
        {\rm P6}: &\min_{\{\psi^{\rm Pin}\}}\;\; \sum\limits_{m=1}^M L_m^{\rm LoS} \label{P60}\\
		&~\mathrm{s.t.} -L_x \leq x_{n,m}^{\rm Pin} \leq L_x, \forall n,m. \label{P61}
    \end{align}
\end{subequations}
We can further transfer \eqref{P60} as follows
\begin{equation}
        \begin{split}
            \sum\limits_{m=1}^M L_m^{\rm LoS} &= 20\log \left(\prod_{m=1}^M \left( \frac{4\pi f_c |\psi_m - \psi^{\rm Pin}|}{c} \right) \right)\\
                &\overset{(d)}{\leq} 10 M\log \left( \frac{16\pi^2 f_c^2 \sum_{m=1}^M |\psi_m - \psi^{\rm Pin}|^2}{Mc^2} \right), \label{total path loss single NOMA}
        \end{split}
    \end{equation}
where ${d}$ is obtained by arithmetic mean-geometric mean inequality. From \eqref{total path loss single NOMA}, it is noted that minimizing the total path loss is equivalent to minimizing the sum of the square of distance between the pinching antenna and users. By substituting the distance with coordinates, ${\rm P6}$ can be recast as follows:
\begin{subequations}\label{Prob7}
    \begin{align}
        {\rm P7}: &\min_{\{x^{\rm Pin}\}}\;\; \sum\limits_{m=1}^M (x^{\rm Pin} - x_m)^2 + d_m \label{P70}\\
		&~\mathrm{s.t.} \quad -L_x \leq x^{\rm Pin} \leq L_x, \label{P71}
    \end{align}
\end{subequations}
where $d_m = d^2 + y_m^2$. It is noted that ${\rm P7}$ is convex and the optimal solution is provided by the following proposition. 
\begin{prop}
    The optimal solution of ${\rm P7}$ is given by
    \begin{equation}
        x^{* {\rm Pin}} = \frac{\sum_{m=1}^M x_m}{M}. \label{optimal pin position single NOMA}
    \end{equation}
\end{prop}
\begin{proof}
    Due to \eqref{P71} is convex, the optimal solution can be obtained by letting the first-order derivative of \eqref{P71} equal to $0$.
\end{proof}
\subsubsection{Power Allocation for Users}
Given the placement of the pinching antenna, the channels between the pinching antenna and users are also determined. As a result, the SIC decoding order can be further determined by sorting the channel gains. Without loss of generality, we assume the channel gains are sorted as $|h_1|^2 \leq \cdots \leq |h_M|^2$. Then, we have
\begin{equation}
    R_m = R_{m,m} = \log_2 \left(1+ \frac{|h_m|^2 p_m}{\sum_{j=m+1}^M |h_m|^2 p_j + \sigma^2} \right).
\end{equation}
According to the derivation of \cite{liu2019placement,wu2023two}, the optimal power allocation scheme is achieved to maximize the sum rate when each user is first allocated the minimum power required to meet its rate requirement and then remaining power is assigned to the user with the best channel condition. Hence, the optimal power allocation is given by
\begin{equation}
    \begin{cases}
        p^*_1 = \frac{2^{R_t}-1}{2^{R_t}} \left(P_{max} + \frac{\sigma^2}{|h_1|^2} \right),\\
        p^*_2 = \frac{2^{R_t}-1}{2^{R_t}} \left(P_{max} - p^*_1 + \frac{\sigma^2}{|h_2|^2} \right),\\
        \vdots \\
        p^*_{M-1} = \frac{2^{R_t}-1}{2^{R_t}} \left(P_{max} - \sum\limits_{m=1}^{M-2} p^*_{m} + \frac{\sigma^2}{|h_{M-1}|^2} \right),\\
        p^*_M = P_{max} -  \sum\limits_{m=1}^{M-1} p^*_{m}.
    \end{cases} \label{optimal power allocation single NOMA}
\end{equation}
As a result, the sum rate can be expressed as follows:
\begin{equation}
    \sum\limits_{m=1}^M R_m = (M-1)R_t + \log_2\left(1+\frac{|h_M|^2 p^*_M}{\sigma^2}\right). 
\end{equation}
\begin{rem}
If there are $N$ pinching antennas in NOMA scheme, the channel gain of the $m-$th user becomes $h_m = \sum\limits_{n=1}^N \frac{\eta^{\frac{1}{2}} e^{-j\frac{2\pi}{\lambda} \left|\psi_m - \psi_{n}^{\rm Pin}\right|}}{\left|\psi_m - \psi_{n}^{\rm Pin}\right|} e^{-j \theta_{n,m}}$ \cite{ding2024flexible}. As can be observed, each user’s channel gain contains sums of complex numbers, it is difficult to derive a closed-form of the placement of pinching antennas by considering user interference and SIC. As a result, a sophisticated optimization algorithm is required, which will be our further research direction.
\end{rem}
\vspace{-0.5cm}
\section{Simulation Results}
In this section, the simulation results demonstrate the superior performance of the pinching-antenna system compared to the conventional fixed-antenna system. The parameters are set as follows: $f_c = 28$ GHz, $d = 3$m, and $\sigma^2 = -90$ dBm. It is assumed that there are 4 users in the service area. In addition, the conventional fixed-antenna system is the benchmark, which is assumed to be deployed at $(0,0,d)$. 
\begin{figure}[t]
     \centering
     \includegraphics[width=0.4\textwidth]{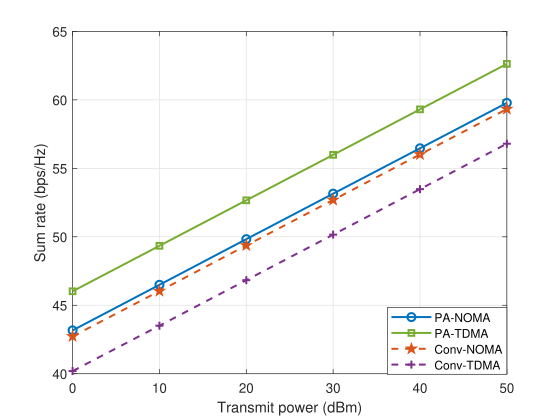}
     \caption{The transmit power versus the sum rate, where only one pinching antenna on the waveguide.}
     \label{singleantenna}
\vspace{-0.4cm}
\end{figure}

\begin{figure}[t]
     \centering
     \includegraphics[width=0.4\textwidth]{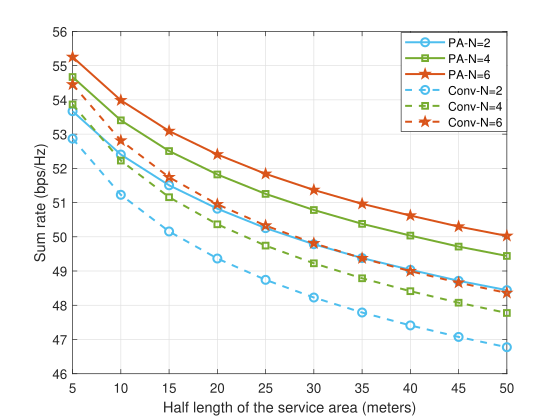}
     \caption{The service area versus the sum rate, where multiple pinching antennas, $N = 2,4,6$, on the waveguide.}
     \label{multiantennas}
\vspace{-0.5cm}
\end{figure}
Fig. \ref{singleantenna} shows the sum rate performance of the pinching-antenna system, where there is only one pinching antenna on a waveguide. The two half side lengths of the service area is set as $L_x = 60$m and $L_y = 5$m, respectively.  It is noted that the sum rate increases as the transmit power increases for all schemes. Two multiple access schemes are investigated, namely TDMA and NOMA. Fig. \ref{singleantenna} shows that the pinching-antenna system can achieve a higher sum rate than the conventional fixed-antenna system in both TDMA and NOMA schemes. However, as can be observed, the performance gain introduced by pinching-antenna systems in NOMA is not obvious as that in TDMA when single pinching antenna is deployed. The reason is that the pinching antenna can be moved to the closet position corresponding to a specific user in each time slot when TDMA is employed. Hence, the large-scale path loss is maximally mitigated. In contrast, the pinching antenna is deployed at a pre-designed position and will not change the position when NOMA is employed. As a result, the large-scale path loss decreases the channel quality of the users that are far away from the antenna. \par
Fig. \ref{multiantennas} shows the sum rate performance of the pinching-antenna system versus the service area, where there are multiple pinching antennas on a waveguide. The service area in this simulation is assumed to be a square with $L_x = L_y$. In this figure, only TDMA is investigated. It is noted that the sum rate increases as the transmit power increases for all schemes and the pinching-antenna system can still achieve a higher sum rate than the conventional fixed-antenna system when there are multiple antennas. In addition, as can be observed, the sum rate increases as the number of antennas increases. As a result, the sum rate performance can benefit from multiple antenna systems. 
\vspace{-0.4cm}
\section{Conclusion} 
The paper proposed a low-complexity placement design of pinching-antenna systems to maximize the sum rate of multiple downlink users. The investigation focused on both TDMA and NOMA schemes when the pinching-antenna system is equipped with a single antenna. Furthermore, the study was extended to scenarios where the pinching-antenna system is equipped with multiple antennas under the TDMA scheme. This paper derived the closed-form expression of the placement of pinching antennas to avoid solving complex optimization problems. Furthermore, the superiority introduced by pinching-antennas systems on sum rate performance compared with conventional fixed-antenna systems was verified by the simulation results.     
\vspace{-0.3cm}
\bibliographystyle{IEEEtran}
\bibliography{EEref}
\end{document}